\documentclass[letterpaper, 10 pt, conference]{ieeeconf}

\IEEEoverridecommandlockouts                              

\usepackage[T1]{fontenc}
\usepackage[latin9]{inputenc}
\usepackage{amsmath}
\usepackage{amssymb}
\usepackage{graphicx}

\usepackage{amsthm}

\makeatletter
  \theoremstyle{plain}
  \newtheorem{lem}{\protect\lemmaname}
  \theoremstyle{definition}
  \newtheorem*{example*}{\protect\examplename}
  \theoremstyle{plain}
  \newtheorem{assumption}{\protect\assumptionname}
\theoremstyle{plain}
\newtheorem{thm}{\protect\theoremname}
  \theoremstyle{remark}
  \newtheorem{rem}{\protect\remarkname}

\usepackage{blkarray}
\usepackage{empheq}
\usepackage{hyperref}
\usepackage[caption=false]{subfig}
\usepackage{svg}
\usepackage{comment}

\graphicspath{{img/}}

\makeatother

\providecommand{\assumptionname}{Assumption}
\providecommand{\examplename}{Example}
\providecommand{\lemmaname}{Lemma}
\providecommand{\remarkname}{Remark}
\providecommand{\theoremname}{Theorem}

\begin{document}
\global\long\def\sign{\text{\text{sign}}}
\global\long\def\R{\mathbb{R}}
\global\long\def\N{\mathbb{Z}}
\global\long\def\tr{^{\top}}
\global\long\def\L{\mathcal{L}}
\global\long\def\K{\mathcal{K}}

\title{\LARGE \bf
Interval Prediction for Continuous-Time Systems with Parametric Uncertainties}

\author{Edouard Leurent, Denis Efimov, Tarek Ra\"issi, Wilfrid Perruquetti
	\thanks{Edouard Leurent is with Renault, France.}
	\thanks{Denis Efimov is with Inria, Univ. Lille, CNRS, UMR 9189 - CRIStAL, F-59000 Lille, France and with ITMO University, 197101 Saint Petersburg,
Russia.}
    \thanks{Tarek Ra\"issi is with Conservatoire National des Arts et M\'etiers (CNAM), Cedric 292, Rue St-Martin, 75141 Paris, France. }
    \thanks{Wilfrid Perruquetti is with Centrale Lille, CNRS, UMR 9189 - CRIStAL,
    	F-59000 Lille, France.}
}

\maketitle
\begin{abstract}
The problem of behaviour prediction for linear parameter-varying systems is considered in the interval framework. It is assumed that the system is subject to uncertain inputs and the vector of scheduling parameters is unmeasurable, but all uncertainties take values in a given admissible set. Then an interval predictor is designed and its stability is guaranteed applying Lyapunov function with a novel structure. The conditions of stability are formulated in the form of linear matrix inequalities. Efficiency of the theoretical results is demonstrated in the application to safe motion planning for autonomous vehicles.
\end{abstract}

\section{Introduction}

There are plenty of emerging application domains nowadays, where the decision algorithms have to operate in the conditions of a severe uncertainty. Therefore, the decision procedures need more information, then the estimation, identification and prediction algorithms come to the attention. In most of these applications, even the nominal simplified models are nonlinear, and in order to solve the problem of estimation and control in nonlinear and uncertain systems, a popular approach is based on the Linear Parameter-Varying (LPV) representation of their dynamics \cite{Shamma2012,Marcos_Balas04,Shamma_Cloutier93,Tan97}, since it allows to reduce the problem to the linear context at the price of augmented parametric variation.

In the presence of uncertainty (unknown parameters or/and external disturbances) the design of a conventional estimator or predictor, approaching the ideal value of the state, can be realized under restrictive assumptions only. However, an interval estimation/prediction remains frequently feasible: using input-output information an algorithm evaluates the set of admissible values (interval) for the state at each instant of time \cite{Efimov2016,Raiessi2018}. The interval length must be minimized via a parametric tuning of the system, and it is typically proportional to the size of the model uncertainty \cite{Chebotarev2015}.

There are many approaches to design interval/set-membership estimators and predictors \cite{Jaulin02,Kieffer_Walter04,Bernard_Gouze04,Moisan_Bernard_Gouze09}, and this paper focuses on the design based on the monotone systems theory \cite{Bernard_Gouze04,Moisan_Bernard_Gouze09,RVZ10,REZ11,Efimov_a2012}.
In such a way the main difficulty for synthesis consists in ensuring cooperativity of the interval error dynamics by a proper design of the algorithm. As it has been shown in \cite{MazencBernard11,REZ11,Combastel2012}, such a complexity of the design can be handled by applying an additional transformation of coordinates to map a stable system into a stable and monotone one, see also \cite{Efimov_a2013,Chebotarev2015}. 

The objective of this paper is to propose an interval predictor for LPV systems. The main difficulty to overcome is the predictor stability, which contrarily to an observer cannot be imposed by a proper design of the gains. An interval inclusion of the uncertain components can be restrictive and transform an initially stable system to an unstable one. To solve this problem, a generic predictor is proposed for an LPV system. To analyze stability of the predictor, which is modeled by a nonlinear Lipschitz dynamics, a novel non-conservative Lyapunov function is developed, whose features can be verified through solution of linear matrix inequalities (LMIs). Finally, the utility of the developed theory is demonstrated on the problem of robust path planning for a self-driving car by making comparison with earlier results from \cite{Leurent2018}.

\section{\label{sec:Preliminaries} Preliminaries}

We denote the real numbers  by $\mathbb{R}$, the integers by $\N$, $\R_{+}=\{\tau\in\R:\tau\ge0\}$,  $\N_{+}=\N\cap\R_{+}$ and the sequence of integers $1,...,k$ as $\overline{1,k}$. Euclidean norm for a vector $x\in\mathbb{R}^{n}$ will be denoted as $|x|$, and for a measurable and locally essentially bounded input $u:\mathbb{R}_{+}\to\mathbb{R}$ we denote its $L_{\infty}$ norm by $
||u||_{[t_{0},t_{1}]}=\text{ess}\sup_{t\in[t_{0},t_{1})}|u(t)|
$
. If $t_{1}=+\infty$ then we will simply write $||u||$. We will denote as $\mathcal{L}_{\infty}$ the set of all inputs $u$ with the property $||u||<\infty$. 

The symbols $I_{n}$, $E_{n\times m}$ and $E_{p}$ denote the identity matrix with dimension $n\times n$, and the matrices with all elements equal 1 with dimensions $n\times m$ and $p\times1$, respectively.

For a matrix $A\in\R^{n\times n}$ the vector of its eigenvalues is denoted as $\lambda(A)$, $||A||_{max}=\max_{i=\overline{1,n},j=\overline{1,n}}|A_{i,j}|$ (the elementwise maximum norm, it is not sub-multiplicative) and $||A||_{2}=\sqrt{\max_{i=\overline{1,n}}\lambda_{i}(A\tr A)}$ (the induced $L_{2}$ matrix norm), the relation $||A||_{max}\le||A||_{2}\le n||A||_{max}$ is satisfied between these norms.

\subsection{Interval arithmetic}

For two vectors $x_{1},x_{2}\in\mathbb{R}^{n}$ or matrices $A_{1},A_{2}\in\R^{n\times n}$, the relations $x_{1}\le x_{2}$ and $A_{1}\le A_{2}$ are understood elementwise. The relation $P\prec0$ ($P\succ0$) means that a symmetric matrix $P\in\R^{n\times n}$ is negative (positive) definite. Given a matrix $A\in\R^{m\times n}$, define $A^{+}=\max\{0,A\}$, $A^{-}=A^{+}-A$ (similarly for vectors) and denote the matrix of absolute values of all elements by $|A|=A^{+}+A^{-}$. 
\begin{lem}
\textup{\cite{EFRZS12}} \label{lem:interval} Let $x\in\mathbb{R}^{n}$ be a vector variable, $\underline{x}\le x\le\overline{x}$ for some $\underline{x},\overline{x}\in\mathbb{R}^{n}$. 

\textup{(1)} If $A\in\R^{m\times n}$ is a constant matrix, then
\begin{equation}
A^{+}\underline{x}-A^{-}\overline{x}\le Ax\le A^{+}\overline{x}-A^{-}\underline{x}.\label{eq:Interval1}
\end{equation}

\textup{(2)} If $A\in\R^{m\times n}$ is a matrix variable and \textup{$\underline{A}\le A\le\overline{A}$} for some $\underline{A},\overline{A}\in\R^{m\times n}$, then
\begin{gather}
\underline{A}^{+}\underline{x}^{+}-\overline{A}^{+}\underline{x}^{-}-\underline{A}^{-}\overline{x}^{+}+\overline{A}^{-}\overline{x}^{-}\leq Ax\label{eq:Interval2}\\
\leq\overline{A}^{+}\overline{x}^{+}-\underline{A}^{+}\overline{x}^{-}-\overline{A}^{-}\underline{x}^{+}+\underline{A}^{-}\underline{x}^{-}.\nonumber 
\end{gather}
\end{lem}
Furthermore, if $-\overline{A}=\underline{A}\le0\le\overline{A}$, then the inequality (\ref{eq:Interval2}) can be simplified: $-\overline{A}(\overline{x}^{+}+\underline{x}^{-})\leq Ax\leq\overline{A}(\overline{x}^{+}+\underline{x}^{-})$.

\subsection{Nonnegative systems}

A matrix $A\in\R^{n\times n}$ is called Hurwitz if all its eigenvalues have negative real parts, it is called Metzler if all its elements outside the main diagonal are nonnegative. Any solution of the linear system
\begin{gather}
\dot{x}(t)=Ax(t)+B\omega(t),\:t\geq0,\label{eq:LTI_syst}\\
y(t)=Cx(t)+D\omega(t),\nonumber 
\end{gather}
with $x(t)\in\R^{n}$, $y(t)\in\R^{p}$ and a Metzler matrix $A\in\R^{n\times n}$, is elementwise nonnegative for all $t\ge0$ provided that $x(0)\ge0$, $\omega:\R_{+}\to\R_{+}^{q}$ and $B\in\R_{+}^{n\times q}$ \cite{FarinaRinaldi2000,Smith95}. The output solution $y(t)$ is nonnegative if $C\in\R_{+}^{p\times n}$ and $D\in\R_{+}^{p\times q}$. Such dynamical systems are called cooperative (monotone) or nonnegative if only initial conditions in $\R_{+}^{n}$ are considered \cite{FarinaRinaldi2000,Smith95}.
\begin{lem}
\label{lem:l2}\textup{\cite{REZ11}} Given the matrices $A\in\R^{n\times n}$, $Y\in\R^{n\times n}$ and \textup{$C\in\R^{p\times n}$. }If there is a matrix \textup{$L\in\R^{n\times p}$} such that the matrices $A-LC$ and $Y$ have the same eigenvalues, then there is a matrix $S\in\R^{n\times n}$ such that $Y=S(A-LC)S^{-1}$ provided that the pairs $(A-LC,\chi_{1})$ and $(Y,\chi_{2})$ are observable for some $\chi_{1}\in\R^{1\times n}$, $\chi_{2}\in\R^{1\times n}$.\textup{ }
\end{lem}
This result allows to represent the system \eqref{eq:LTI_syst} in its nonnegative form via a similarity transformation of coordinates.
\begin{lem}
\label{lem:l3}\textup{\cite{Efimov_a2013}} Let $D\in\Xi\subset\R^{n\times n}$ be a matrix variable satisfying the interval constraints $\Xi=\{D\in\R^{n\times n}:\,D_{a}-\Delta\le D\le D_{a}+\Delta\}$ for some $D_{a}^{\text{T}}=D_{a}\in\R^{n\times n}$ and $\Delta\in\R_{+}^{n\times n}$. If for some constant $\mu\in\R_{+}$ and a diagonal matrix $\Upsilon\in\R^{n\times n}$ the Metzler matrix $Y=\mu E_{n\times n}-\Upsilon$ has the same eigenvalues as the matrix $D_{a}$, then there is an orthogonal matrix $S\in\R^{n\times n}$ such that the matrices $S^{\text{T}}DS$ are Metzler for all $D\in\Xi$ provided that $\mu>n||\Delta||_{max}$.\textup{ }
\end{lem}
In the last lemma, the existence of similarity transformation is proven for an interval of matrices, \emph{e.g}. LPV dynamics.

\section{\label{sec:Problem-statement} Problem statement}

Consider an LPV system:
\begin{equation}
\dot{x}(t)=A(\theta(t))x(t)+Bd(t),\;t\geq0,\label{eq:LPV_syst}
\end{equation}
where $x(t)\in\R^{n}$ is the state, $\theta(t)\in\Pi\subset\R^{r}$ is the vector of scheduling parameters with a known set of admissible values $\Pi$, $\theta\in\L_{\infty}^{r}$; the signal $d:\R_{+}\to\R^{m}$ is the external input. The values of the scheduling vector $\theta(t)$ are not available for measurements, and only the set of admissible values $\Pi$ is known. The matrix $B\in\R^{n\times m}$ is known, the matrix function $A:\Pi\to\R^{n\times n}$ is locally bounded (continuous) and known.

The following assumptions will be used in this work.
\begin{assumption}
\label{ass:a1} In the system \eqref{eq:LPV_syst}, $x\in\L_{\infty}^{n}$. In addition, $x(0)\in[\underline{x}_{0},\overline{x}_{0}]$ for some known $\underline{x}_{0},\overline{x}_{0}\in\R^{n}$.
\end{assumption}

\begin{assumption}
\label{ass:a2} There exists known signals $\underline{d},\overline{d}\in\L_{\infty}^{n}$ such that $\underline{d}(t)\leq d(t)\leq\overline{d}(t)$ for all $t\geq0$.
\end{assumption}
Assumption \ref{ass:a1} means that the system \eqref{eq:LPV_syst} generates stable trajectories with a bounded state $x$ for the applied class of inputs $d$, and the initial conditions $x(0)$ are constrained to belong to a given interval $[\underline{x}_{0},\overline{x}_{0}]$. In Assumption \ref{ass:a2}, it is supposed that the input $d(t)$ belongs to a known bounded interval $[\underline{d}(t),\overline{d}(t)]$ for all $t\in\R_{+}$, which is the standard hypothesis for the interval estimation \cite{Efimov2016,Raiessi2018}.

Note that since the function $A$ and the set $\Pi$ are known, and $\theta\in\Pi$, then there exist matrices $\underline{A},\overline{A}\in\R^{n\times n}$, which can be easily computed, such that 
\[
\underline{A}\leq A(\theta)\leq\overline{A},\quad\forall\theta\in\Pi.
\]

\subsection{The goal}

The objective of this work is to design an \emph{interval predictor} for the system \eqref{eq:LPV_syst}, which takes the information on the initial conditions $[\underline{x}_{0},\overline{x}_{0}]$, the admissible bounds on the values of the exogenous input $[\underline{d}(t),\overline{d}(t)]$, the information about $A$ and $\Pi$ (\emph{e.g}. the matrices $\underline{A},\overline{A}$, but not the instant value of $\theta(t)$) and generates bounded interval estimates $\underline{x}(t),\overline{x}(t)\in\R^{n}$ such that
\begin{equation}
\underline{x}(t)\leq x(t)\leq\overline{x}(t),\quad\forall t\geq0.\label{eq:Interval_Goal}
\end{equation}

\subsection{A motivation example}

Following the result of Lemma \ref{lem:interval}, there is a straightforward solution to the problem used in \cite{Leurent2018}:
\begin{eqnarray}
\dot{\underline{x}}(t) & = & \underline{A}^{+}\underline{x}^{+}(t)-\overline{A}^{+}\underline{x}^{-}(t)-\underline{A}^{-}\overline{x}^{+}(t)\nonumber \\
 &  & +\overline{A}^{-}\overline{x}^{-}(t)+B^{+}\underline{d}(t)-B^{-}\overline{d}(t),\label{eq:IP_direct}\\
\dot{\overline{x}}(t) & = & \overline{A}^{+}\overline{x}^{+}(t)-\underline{A}^{+}\overline{x}^{-}(t)-\overline{A}^{-}\underline{x}^{+}(t)\nonumber \\
 &  & +\underline{A}^{-}\underline{x}^{-}(t)+B^{+}\overline{d}(t)-B^{-}\underline{d}(t),\nonumber \\
 &  & \underline{x}(0)=\underline{x}_{0},\;\overline{x}(0)=\overline{x}_{0},\nonumber 
\end{eqnarray}
then it is obvious to verify that the relations \eqref{eq:Interval_Goal} are satisfied, but the stability analysis of the system \eqref{eq:IP_direct} is more tricky. Indeed, \eqref{eq:IP_direct} is a purely nonlinear system (since $\underline{x}^{+}$, $\underline{x}^{-}$, $\overline{x}^{+}$ and $\overline{x}^{-}$ are globally Lipschitz functions of the state $\underline{x}$ and $\overline{x}$), whose robust stability with respect to the bounded external inputs $\underline{d}$ and $\overline{d}$ can be assessed if a suitable Lyapunov function is found. And it is easy to find an example, where the matrices $\underline{A}$ and $\overline{A}$ are stable, but the system \eqref{eq:IP_direct} is not:
\begin{example*}
[motivating] Consider a scalar system:
\[
\dot{x}(t)=-\theta(t)x(t)+d(t),\;t\geq0,
\]
where $x(t)\in\R$ with $x(0)\in[\underline{x}_{0},\overline{x}_{0}]=[1.0, 1.1]$, $\theta(t)\in\Pi=[\underline{\theta},\overline{\theta}]=[0.5,1.5]$ and $d(t)\in[\underline{d},\overline{d}]=[-0.1,0.1]$ for all $t\geq0$. Obviously, assumptions \ref{ass:a1} and \ref{ass:a2} are satisfied, and this uncertain dynamics produces bounded trajectories (to prove this consider a Lyapunov function $V(x)=x^{2}$). Then the interval predictor \eqref{eq:IP_direct} takes the form:
\begin{eqnarray*}
\dot{\underline{x}}(t) & = & -\overline{\theta}\overline{x}^{+}(t)+\underline{\theta}\overline{x}^{-}(t)+\underline{d},\\
\dot{\overline{x}}(t) & = & -\underline{\theta}\underline{x}^{+}(t)+\overline{\theta}\underline{x}^{-}(t)+\overline{d}.
\end{eqnarray*}
The results of simulation are shown in Fig. \ref{fig:IP_Direct}. As we can conclude, additional consideration and design are needed to properly solve the posed problem.
\begin{figure}
\begin{centering}
\includegraphics[width=\linewidth]{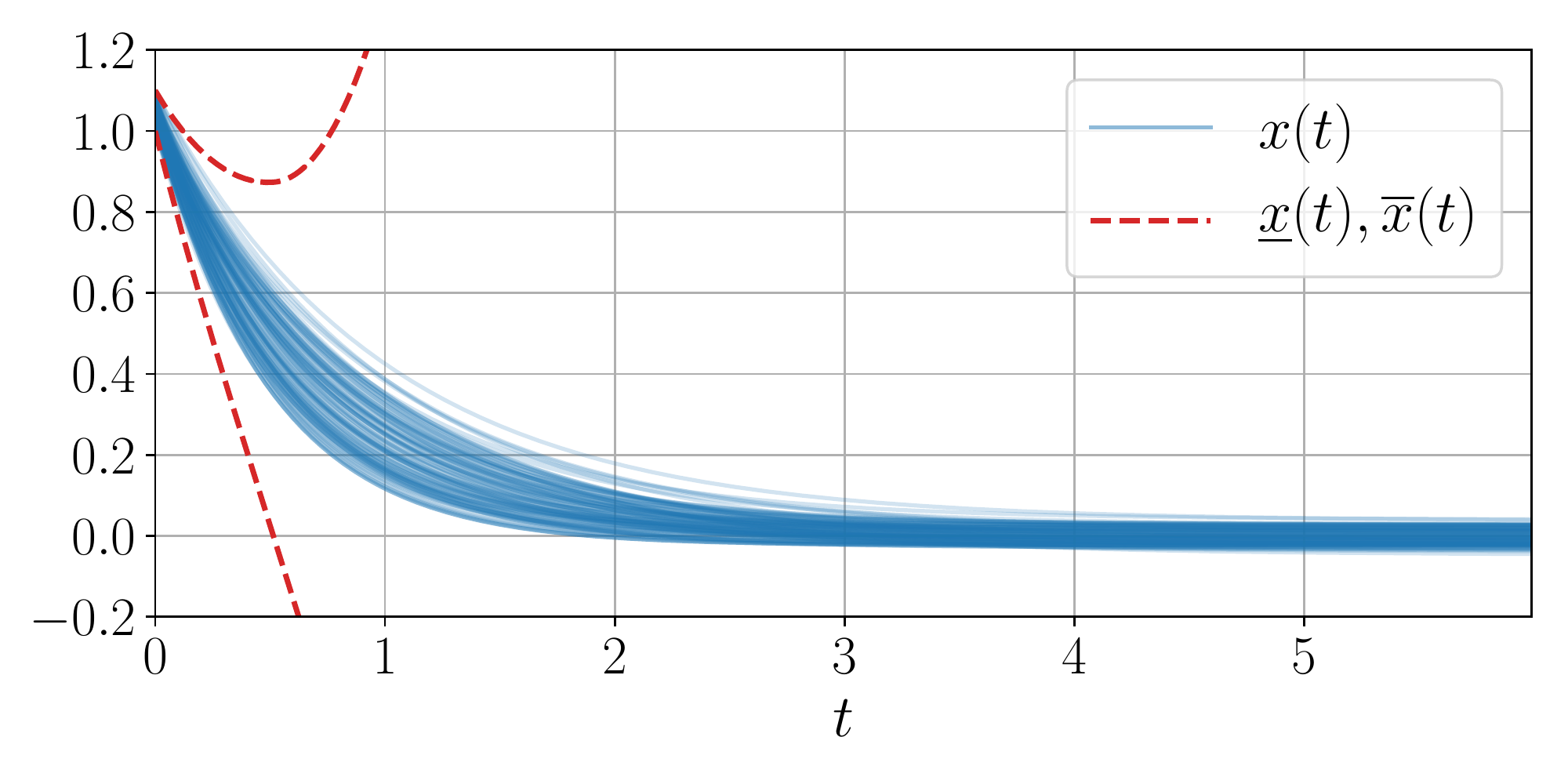}
\par\end{centering}
\caption{\label{fig:IP_Direct} The results of prediction by \eqref{eq:IP_direct}: even in such a simplistic setting, the predictor is unstable and diverges quickly.}
\end{figure}
\end{example*}

\section{\label{sec:Main-results} Interval predictor design}

Note that, in related papers \cite{AitRami2008,RVZ10,Bolajraf2011,Efimov_a2013,Efimov_tac2013,Chebotarev2015}, various interval observers for LPV systems have been proposed, but in those works the cooperativity and stability of the estimation error dynamics are ensured by a proper selection of observer gains and/or by design of control algorithms, which can be dependent on $\underline{x}$, $\overline{x}$ and guarantee the observer robust stability. For interval predictor there is no such a freedom, then a careful selection of hypotheses has to be made in order to provide a desired solution.
We will additionally assume the following:
\begin{assumption}
\label{ass:a3} There exist a Metzler matrix $A_{0}\in\R^{n\times n}$ and matrices $\Delta A_{i}\in\R^{n\times n}$, $i=\overline{1,N}$ for some $N\in\N_{+}$ such that the following relations are satisfied for all $\theta\in\Pi$:
\begin{gather*}
A(\theta)=A_{0}+\sum_{i=1}^{N}\lambda_{i}(\theta)\Delta A_{i},\\
\sum_{i=1}^{N}\lambda_{i}(\theta)=1;\;\lambda_{i}(\theta)\in[0,1],\;i=\overline{1,N}.
\end{gather*}
\end{assumption}
Therefore, it is assumed that the matrix $A(\theta)$ for any $\theta\in\Pi$ can be embedded in a polytope defined by $N$ known vertices $\Delta A_{i}$ with the given center $A_{0}$, which admits already useful properties. According to the results of lemmas \ref{lem:l2} and \ref{lem:l3}, the fulfillment of Assumption \ref{ass:a3} can be imposed by applying a properly designed similarity transformation, which maps a matrix (interval of matrices) to a Metzler one. Design of such a transformation is not considered in this work, and we will just suppose in Assumption \ref{ass:a3} that the system \eqref{eq:LPV_syst} has been already put in the right form:
\[
\dot{x}(t)=[A_{0}+\sum_{i=1}^{N}\lambda_{i}(\theta(t))\Delta A_{i}]x(t)+Bd(t).
\]
Denote
\[
\Delta A_{+}=\sum_{i=1}^{N}\Delta A_{i}^{+},\;\Delta A_{-}=\sum_{i=1}^{N}\Delta A_{i}^{-},
\]
then the following interval predictor can be designed:
\begin{thm}
\label{thm:main} Let assumptions \ref{ass:a1}\textendash \ref{ass:a3} be satisfied for the system \eqref{eq:LPV_syst}, then an interval predictor
\begin{eqnarray}
\dot{\underline{x}}(t) & = & A_{0}\underline{x}(t)-\Delta A_{+}\underline{x}^{-}(t)-\Delta A_{-}\overline{x}^{+}(t)\nonumber \\
 &  & +B^{+}\underline{d}(t)-B^{-}\overline{d}(t),\label{eq:IP_main}\\
\dot{\overline{x}}(t) & = & A_{0}\overline{x}(t)+\Delta A_{+}\overline{x}^{+}(t)+\Delta A_{-}\underline{x}^{-}(t)\nonumber \\
 &  & +B^{+}\overline{d}(t)-B^{-}\underline{d}(t),\nonumber \\
 &  & \underline{x}(0)=\underline{x}_{0},\;\overline{x}(0)=\overline{x}_{0}\nonumber 
\end{eqnarray}
ensures the property \eqref{eq:Interval_Goal}. If there exist diagonal matrices $P$, $Q$, $Q_{+}$, $Q_{-}$, $Z_{+}$, $Z_{-}$, $\Psi_{+}$, $\Psi_{-}$, $\Psi$, $\Gamma\in\R^{2n\times2n}$ such that the following LMIs are satisfied:
\begin{gather*}
P+\min\{Z_{+},Z_{-}\}>0,\;\Upsilon\preceq0,\;\Gamma>0,\\
Q+\min\{Q_{+},Q_{-}\}+2\min\{\Psi_{+},\Psi_{-}\}>0,
\end{gather*}
where{\footnotesize{}
\begin{gather*}
\Upsilon=\left[\begin{array}{cccc}
\Upsilon_{11} & \Upsilon_{12} & \Upsilon_{13} & P\\
\Upsilon_{12}^{\top} & \Upsilon_{22} & \Upsilon_{23} & Z_{+}\\
\Upsilon_{13}^{\top} & \Upsilon_{23}^{\top} & \Upsilon_{33} & Z_{-}\\
P & Z_{+} & Z_{-} & -\Gamma
\end{array}\right],\\
\Upsilon_{11}=\mathcal{A}^{\top}P+P\mathcal{A}+Q,\;\Upsilon_{12}=\mathcal{A}^{\top}Z_{+}+PR_{+}+\Psi_{+},\\
\Upsilon_{13}=\mathcal{A}^{\top}Z_{-}+PR_{-}+\Psi_{-},\;\Upsilon_{22}=Z_{+}R_{+}+R_{+}^{\top}Z_{+}+Q_{+},\\
\Upsilon_{23}=Z_{+}R_{-}+R_{+}^{\top}Z_{-}+\Psi,\;\Upsilon_{33}=Z_{-}R_{-}+R_{-}^{\top}Z_{-}+Q_{-},\\
\mathcal{A}=\left[\begin{array}{cc}
A_{0} & 0\\
0 & A_{0}
\end{array}\right],\;R_{+}=\left[\begin{array}{cc}
0 & -\Delta A_{-}\\
0 & \Delta A_{+}
\end{array}\right],\;R_{-}=\left[\begin{array}{cc}
\Delta A_{+} & 0\\
-\Delta A_{-} & 0
\end{array}\right],
\end{gather*}
}then the predictor \eqref{eq:IP_main} is input-to-state stable with respect to the inputs $\underline{d}$, $\overline{d}$.
\end{thm}

Note the requirement that the matrix $P$ has to be diagonal is not restrictive, since for a Metzler matrix $\mathcal{A}$, its stability is equivalent to existence of a diagonal solution $P$ of the Lyapunov equation $\mathcal{A}^{\top}P+P\mathcal{A}\prec0$ \cite{FarinaRinaldi2000}.

\begin{proof}
	First, let us demonstrate \eqref{eq:Interval_Goal}, to this end note that
	\[
	-\Delta A_{+}\underline{x}^{-}-\Delta A_{-}\overline{x}^{+}\leq\sum_{i=1}^{N}\lambda_{i}\Delta A_{i}x\leq\Delta A_{+}\overline{x}^{+}+\Delta A_{-}\underline{x}^{-}
	\]
	and introducing usual interval estimation errors $\underline{e}=x-\underline{x}$ and $\overline{e}=\overline{x}-x$ and calculating their dynamics we get:
	\begin{eqnarray*}
		\dot{\underline{e}}(t) & = & A_{0}\underline{e}(t)+\underline{r}_{1}(t)+\underline{r}_{2}(t),\\
		\dot{\overline{e}}(t) & = & A_{0}\overline{e}(t)+\overline{r}_{1}(t)+\overline{r}_{2}(t),
	\end{eqnarray*}
	where
	\begin{gather*}
	\underline{r}_{1}=\sum_{i=1}^{N}\lambda_{i}\Delta A_{i}x+\Delta A_{+}\underline{x}^{-}+\Delta A_{-}\overline{x}^{+},\\
	\underline{r}_{2}=Bd-B^{+}\underline{d}+B^{-}\overline{d},\\
	\overline{r}_{1}=\Delta A_{+}\overline{x}^{+}+\Delta A_{-}\underline{x}^{-}-\sum_{i=1}^{N}\lambda_{i}\Delta A_{i}x,\\
	\overline{r}_{2}=B^{+}\overline{d}-B^{-}\underline{d}-Bd.
	\end{gather*}
	Non-negativity or $\underline{r}_{2}$ and $\overline{r}_{2}$ follows from Assumption \ref{ass:a2} and Lemma \ref{lem:interval}. The signals $\underline{r}_{1}$ and $\overline{r}_{1}$ are also nonnegative provided that \eqref{eq:Interval_Goal} holds and due to the calculations above. Note that the relations \eqref{eq:Interval_Goal} are satisfied for $t=0$ by construction and Assumption \ref{ass:a1}, then it is possible to show that $\underline{e}(t)\geq0$ and $\overline{e}(t)\geq0$ for all $t\geq0$ \cite{Smith95}, which confirms the relations \eqref{eq:Interval_Goal}.
	
	Second, let us consider the stability of \eqref{eq:IP_main}, and for this purpose define the extended state vector as $X=[\underline{x}^{\top}\;\;\overline{x}^{\top}]^{\top}$, whose dynamics admit the differential equation:
	\[
	\dot{X}(t)=\mathcal{A}X(t)+R_{+}X^{+}(t)-R_{-}X^{-}(t)+\delta(t),
	\]
	where
	\begin{gather*}
	\delta(t)=\left[\begin{array}{cc}
	-B^{-} & B^{+}\\
	B^{+} & -B^{-}
	\end{array}\right]\left[\begin{array}{c}
	\overline{d}(t)\\
	\underline{d}(t)
	\end{array}\right]
	\end{gather*}
	is a bounded input vector, whose norm is proportional to $\underline{d}$,
	$\overline{d}$. Following \cite{Efimov_CDC2019}, consider a candidate Lyapunov function:
	\begin{gather*}
	V(X)=X^{\top}PX+X{}^{\top}Z_{+}X^{+}-X^{\top}Z_{-}X^{-},
	\end{gather*}
	which is positive definite provided that
	\[
	P+\min\{Z_{+},Z_{-}\}>0,
	\]
	and whose derivative for the system dynamics takes the form, if $\Upsilon\preceq0$:
	\begin{eqnarray*}
		\dot{V} & \leq & -X^{\top}\Omega X+\delta^{\top}\Gamma\delta,
	\end{eqnarray*}
	where
	\[
	\Omega=Q+\min\{Q_{+},Q_{-}\}+2\min\{\Psi_{+},\Psi_{-}\}>0
	\]
	is a diagonal matrix. The substantiated properties of $V$ and its derivative imply that \eqref{eq:IP_main} is input-to-state stable \cite{Khalil2002} with respect to the input $\delta$ (or, by its definition, with respect to $(\underline{d},\overline{d})$).
\end{proof}

\begin{rem}
The LMIs of the above theorem are not conservative, since the restriction on positive definiteness of involved matrix variables is not imposed on all of them separately, but on their combinations:
\begin{gather*}
P+\min\{Z_{+},Z_{-}\}>0,\;\Gamma>0,\\
Q+\min\{Q_{+},Q_{-}\}+2\min\{\Psi_{+},\Psi_{-}\}>0,
\end{gather*}
then some of them can be sign-indefinite or negative-definite ensuring the fulfillment of the last inequality:
$
\Upsilon\preceq0.
$
\end{rem}

\begin{example*}
[motivating, continue] Let us apply the predictor \eqref{eq:IP_main}
to the motivation example:
\begin{eqnarray*}
\dot{\underline{x}}(t) & = & -\overline{\theta}\underline{x}(t)-(\overline{\theta}-\underline{\theta})\underline{x}^{-}(t)+\underline{d},\\
\dot{\overline{x}}(t) & = & -\overline{\theta}\overline{x}(t)+(\overline{\theta}-\underline{\theta})\overline{x}^{+}(t)+\overline{d},
\end{eqnarray*}
where $A_{0}=-\overline{\theta}$ is chosen, then $\Delta A_{+}=\overline{\theta}-\underline{\theta}$, $\Delta A_{-}=0$ and all conditions of Theorem \ref{thm:main} are verified. The results of simulation are shown in Fig.~\ref{fig:IP_New}. As we can see the new predictor generates very reasonable and bounded estimates. 
\begin{figure}
\begin{centering}
\includegraphics[width=\linewidth]{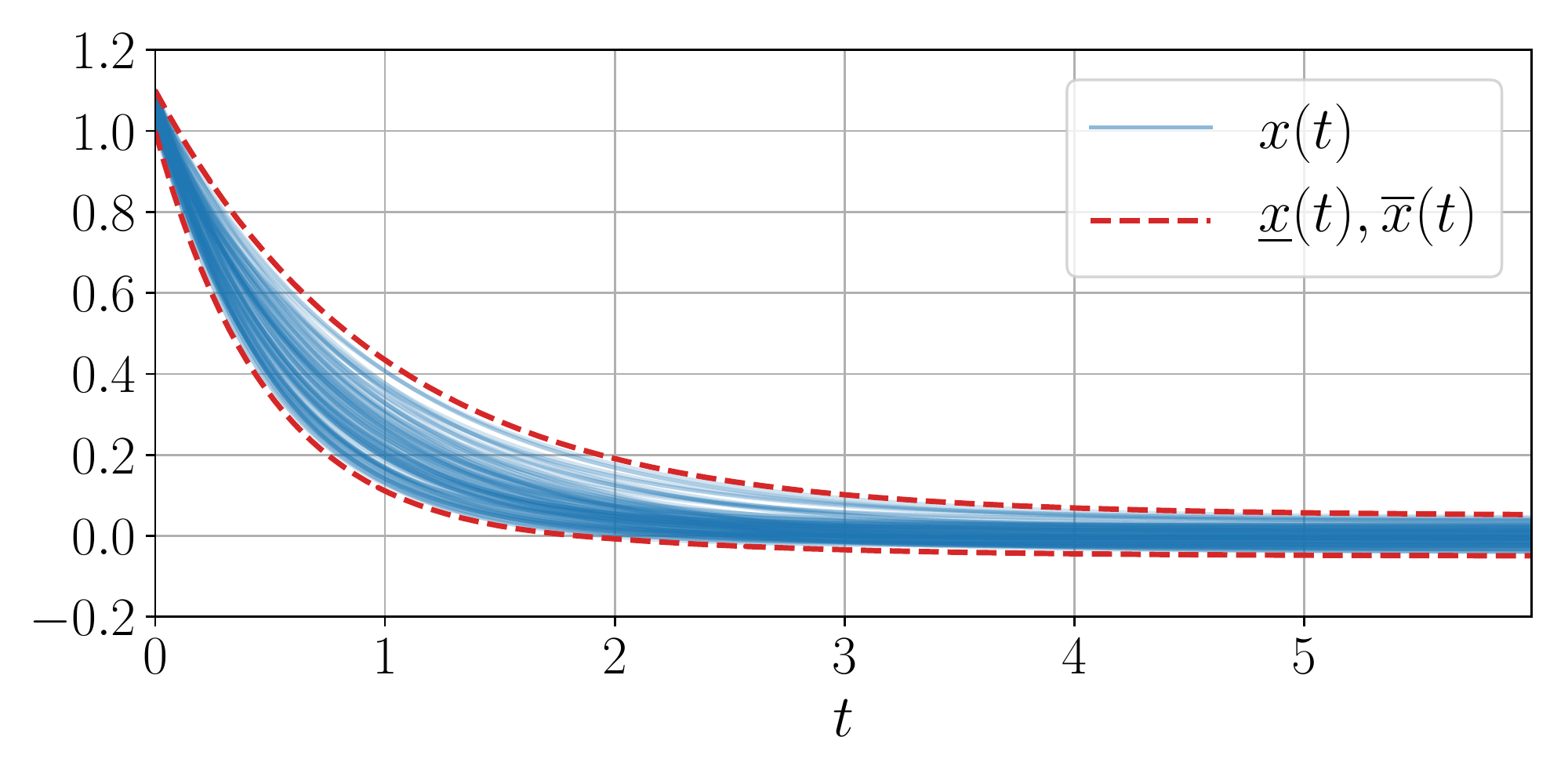}
\par\end{centering}
\caption{\label{fig:IP_New} The results of prediction by \eqref{eq:IP_main}: the new predictor is stable and produces tight bounds.}
\end{figure}
\end{example*}

\section{\label{sec:Examples} Prediction for a self-driving vehicle}

We consider the problem of safe decision-making for autonomous highway driving \cite{highway-env}. The videos and source code of all experiments are available\footnote{\href{https://eleurent.github.io/interval-prediction/}{https://eleurent.github.io/interval-prediction/}}.

An autonomous vehicle is driving on a highway populated with $N$ other agents, and uses Model Predictive Control to plan a sequence of decisions. To that end, it relies on parametrized dynamical model for each agent to predict the future trajectory of each traffic participant: $$\dot{z}_i=f_i(Z,\theta_i),\;i=\overline{1,N},$$ where $f_i$ are described below, $z_i\in\R^4$ is the state of an agent, $Z = [z_1,\dots,z_N]^\top\in\R^{4N}$, and $\theta_i\in\R^5$ is the corresponding vector of unknown parameters. Crucially, this system describes the couplings and interactions between vehicles, so that the autonomous agent can properly anticipate their reactions. 
However, we assume that we do not have access to the true values of the behavioural parameters $\theta=[\theta_1,\dots,\theta_N]^\top$; instead, we merely know that these parameters lie in a set of admissible values $\Pi\subset\R^{5N}$. In order to act safely in the face of uncertainty, we follow the framework of robust decision-making: the agent must consider all the possible trajectories in the space of $Z$ that each vehicle can follow in order to take its decisions. In this work, we focus on how to compute these trajectory enclosures through interval prediction. In the following, we describe the system and its associated interval predictor.

\subsection{Kinematics}

The kinematics of any vehicle $i\in\overline{1,N}$ are represented by the Kinematic Bicycle Model \cite{Polack2017}:
\begin{empheq}[left = \empheqlbrace]{align}
	\dot{x}_i &= v_i\cos(\psi_i), &\dot{y}_i &= v_i\sin(\psi_i), \nonumber\\
	\dot{v}_i &= a_i, &\dot{\psi}_i &= \frac{v_i}{l}tan(\beta_i) \nonumber
\end{empheq}
where $(x_i, y_i)$ is the vehicle position, $v_i$ its forward velocity and $\psi_i$ its heading, $l$ its half-length, $a_i$ is the acceleration command and $\beta_i$ is the slip angle at the centre of gravity, used as a steering command, then $z_i=[x_i,y_i,v_i,\psi_i]^\top$.

\subsection{Control}
Longitudinal behaviour is modelled by a linear controller using three features inspired from the intelligent driver model (IDM) \cite{Treiber2000}: a desired velocity, a braking term to drive slower than the front vehicle, and a braking term to respect a safe distance to the front vehicle.
Denoting $f_i$ the index of the front vehicle preceding vehicle $i$, the acceleration command can be presented as follows:
\begin{equation*}
	a_i = \begin{bmatrix}
	\theta_{i,1} & \theta_{i,2} & \theta_{i,3}
	\end{bmatrix} \begin{bmatrix}
		v_0 - v_i \\
		-(v_{f_i}-v_i)^- \\
		-(x_{f_i} - x_i - (d_0 + v_iT))^- \\
	\end{bmatrix},
	\label{eq:theta_a}
\end{equation*}
where $v_0, d_0$ and $T$ respectively denote the speed limit, jam distance and time gap given by traffic rules.

The lane $L_i$ with the lateral position $y_{L_i}$ and heading $\psi_{L_i}$ is tracked by a cascade controller of lateral position and heading $\beta_i$, which is selected in a way the closed-loop dynamics take the form:

\begin{empheq}[left = \empheqlbrace]{align}
	\label{eq:heading-command}
    \dot{\psi}_i &= \theta_{i,5}\left(\psi_{L_i}+\sin^{-1}\left(\frac{\tilde{v}_{i,y}}{v_i}\right)-\psi_i\right),\\
    \tilde{v}_{i,y} &= \theta_{i,4} (y_{L_i}-y_i). \nonumber
\end{empheq}
We assume that the drivers choose their steering command $\beta_i$ such that \eqref{eq:heading-command} is always achieved: $\beta_i = \tan^{-1}(\frac{l}{v_i}\dot{\psi}_i)$.

\subsection{LPV formulation}

The system presented so far is non-linear and must be cast into the LPV form. We approximate the non-linearities induced by the trigonometric operators through equilibrium linearisation around $y_i=y_{L_i}$ and $\psi_i=\psi_{L_i}$.

This yields the longitudinal dynamics: $\dot{x}_i = v_i$ and
\begin{align*}
\dot v_i &= \theta_{i,1} (v_0 - v_i) + \theta_{i,2} (v_{f_i} - v_i) + \theta_{i,3}(x_{f_i} - x_i - d_0 - v_i T),
\end{align*}
where $\theta_{i,2}$ and $\theta_{i,3}$ are set to $0$ whenever the corresponding features are not active.
It can be rewritten in the form $$\dot{Z} = A(\theta)(Z-Z_c) + d.$$ For example, in the case of two vehicles only:
\begin{equation*}
    Z = \begin{bmatrix}
x_i \\
x_{f_i} \\
v_i \\
v_{f_i} \\
\end{bmatrix}
,\quad
Z_c = \begin{bmatrix}
-d_0-v_0 T \\
0 \\
v_0\\
v_0 \\
\end{bmatrix}
,\quad
d = \begin{bmatrix}
v_0 \\
v_0 \\
0\\
0\\
\end{bmatrix}
\end{equation*}

\begin{equation*}
A(\theta)
=
\begin{blockarray}{ccccc}
 & i & f_i & i & f_i \\
\begin{block}{c[cccc]}
i & 0 & 0 & 1 & 0 \\
f_i & 0 & 0 & 0 & 1 \\
i & -\theta_{i,3} & \theta_{i,3} & -\theta_{i,1}-\theta_{i,2}-\theta_{i,3} & \theta_{i,2} \\
f_i & 0 & 0 & 0 & -\theta_{f_i,1} \\
\end{block}
\end{blockarray}
\end{equation*}

The lateral dynamics are in a similar form:
\begin{equation*}
\begin{bmatrix}
\dot{y}_i \\
\dot{\psi}_i \\
\end{bmatrix}
=
\begin{bmatrix}
0 & v_i \\
-\frac{\theta_{i,4} \theta_{i,5}}{v_i} & -\theta_{i,5}
\end{bmatrix}
\begin{bmatrix}
y_i - y_{L_i} \\
\psi_i - \psi_{L_i}
\end{bmatrix}
+
\begin{bmatrix}
v_i\psi_{L_i} \\
0
\end{bmatrix}
\end{equation*}
Here, the dependency in $v_i$ is seen as an uncertain parametric dependency, \emph{i.e.} $\theta_{i,6}=v_i$, with constant bounds assumed for $v_i$ using an overset of the longitudinal interval predictor.

\subsection{Change of coordinates}
In both cases, the obtained polytope centre $A_0$ is non-Metzler.
We use lemmas \ref{lem:l2} and \ref{lem:l3} to compute a similarity transformation of coordinates. Precisely, we ensure that the polytope is chosen so that its centre $A_0$ is diagonalisable having real eigenvalues, and perform an eigendecomposition to compute its change of basis matrix $S$. The transformed system $Z'=S^{-1}(Z-Z_c)$ verifies Assumption~\ref{ass:a3} as required to apply the interval predictor of Theorem~\ref{thm:main}. Finally, the obtained predictor is transformed back to the original coordinates $Z$ by using the interval arithmetic of Lemma~\ref{lem:interval}.

\subsection{Results}

\begin{figure}
    \subfloat[The naive predictor \eqref{eq:IP_direct} quickly diverges]{
    \centering
      \includegraphics[width=0.49\textwidth]{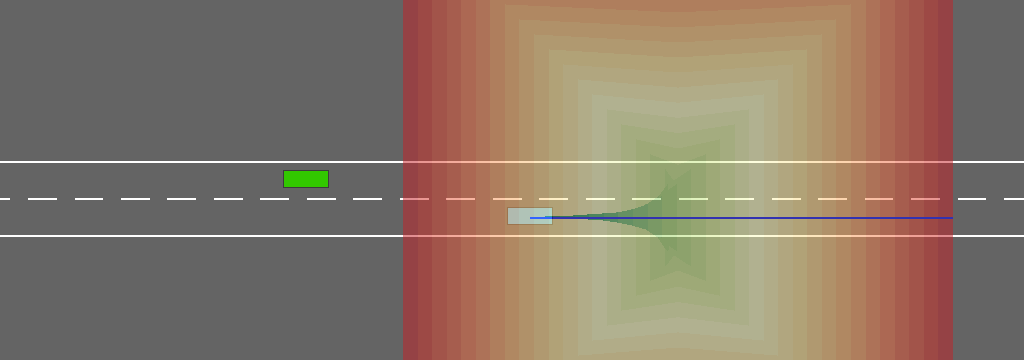}
      \label{sub:hw-a}
    }
    \newline
     \subfloat[The proposed predictor \eqref{eq:IP_main} remains stable]{
    \centering
      \includegraphics[width=0.49\textwidth]{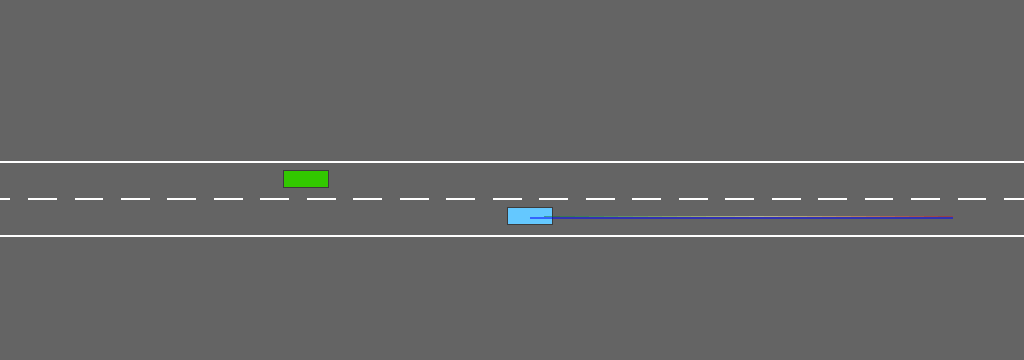}
      \label{sub:hw-b}
}
\newline
    \subfloat[Prediction during a lane change maneuver]{
    \centering
      \includegraphics[width=0.49\textwidth]{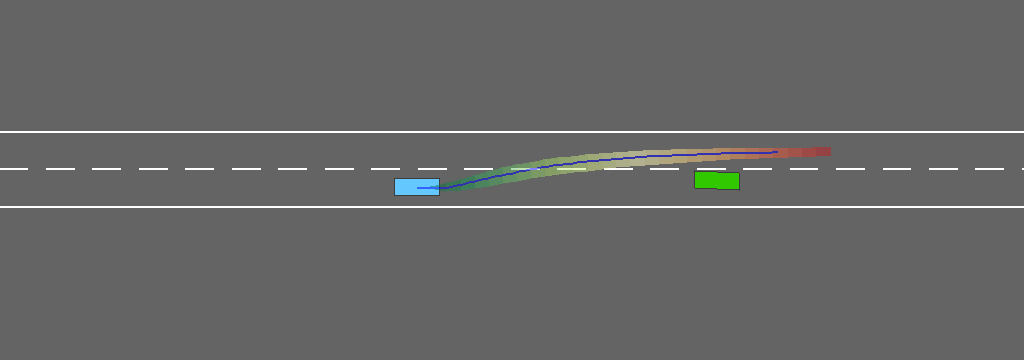}
      \label{sub:hw-c}
}
\newline
    \subfloat[Prediction with uncertainty in the followed lane $L_i$]{
    \centering
      \includegraphics[width=0.49\textwidth]{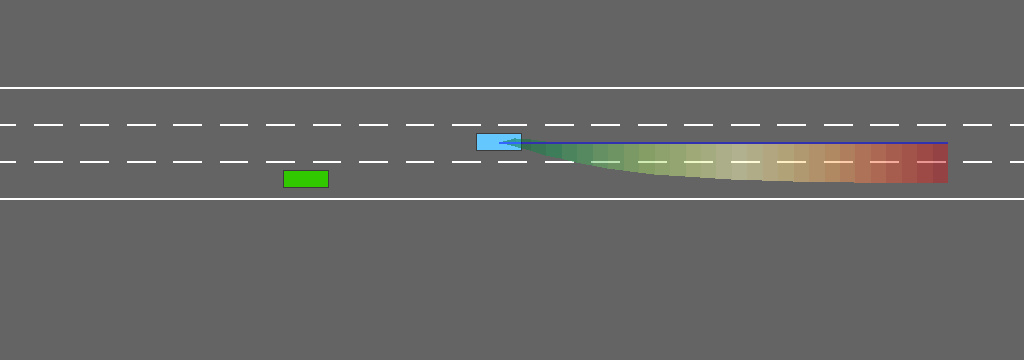}
      \label{sub:hw-d}
}
    \caption{State intervals obtained by the two methods in different conditions.}
    \label{fig:highway}
\end{figure}

We show the resulting intervals in Fig.~\ref{fig:highway}. The target vehicle with uncertain behaviour is in blue, while the ego-vehicle is in green. Its trajectory interval is computed over a duration of two seconds and represented by an area filled with a colour gradient representing time. The ground-truth trajectory is shown in blue. In Fig.~\ref{sub:hw-a}, we observe that the direct predictor \eqref{eq:IP_direct} is unstable and quickly diverges to cover the whole road, thus hindering any sensible decision-making. In \cite{Leurent2018}, they had to circumvent this issue by subdividing $\Pi$ and $[\underline{Z}, \overline{Z}]$ to reduce the initial overestimations and merely delay the divergence \cite{Adrot2003}, at the price of a heavy computational load. In stark contrast, we see in Fig.~\ref{sub:hw-b} that the novel predictor \eqref{eq:IP_main} is stable even over long horizons, which allows the ego-vehicle to plan an overtaking maneuver. Until then, there was little uncertainty in the predicted trajectory for the target vehicle was isolated, but as the ego-vehicle cuts into its lane in Fig.~\ref{sub:hw-c}, we start seeing the  effects of uncertain interactions between the two vehicles, in both longitudinal and lateral directions. Our framework is quite flexible in representing different assumptions on the admissible behaviours. For instance, we show in Fig.~\ref{sub:hw-d} a simulation in which we model a right-hand traffic where drivers are expected to keep to the rightmost lane. In such a situation, it is reasonable to assume that in the absence of any obstacle in front, a vehicle driving on the middle lane will either stay there or return to the right lane, but has no incentive to change to the left-lane. This simple assumption on $L_i$ can easily be incorporated in the interval predictor, and enables the emergence of a realistic behaviour when running the robust decision-making procedure: the ego-vehicle cannot pass another vehicle by its right side, and can only overtake it by its left side. These behaviours displaying safe reasoning under uncertainty are showcased in the attached videos.

\section{Conclusion}

The prediction problem for uncertain LPV systems is solved by designing an interval predictor, which is described by nonlinear differential equations, and whose stability is evaluated using a new Lyapunov function. The corresponding robust stability conditions are expressed in terms of LMIs. The proficiency of the method is demonstrated in application to a problem of safe motion planning for a self-driving car.

\bibliographystyle{ieeetr}
\bibliography{interval_lpv}

\begin{thebibliography}{10}

\bibitem{Shamma2012}
J.~Shamma, {\em Control of Linear Parameter Varying Systems with Applications},
  ch.~An overview of LPV systems, pp.~1--22.
\newblock Springer, 2012.

\bibitem{Marcos_Balas04}
A.~Marcos and J.~Balas, ``Development of linear-parameter-varying models for
  aircraft,'' {\em J. Guidance, Control, Dynamics}, vol.~27, no.~2,
  pp.~218--228, 2004.

\bibitem{Shamma_Cloutier93}
J.~Shamma and J.~Cloutier, ``Gain-scheduled missile autopilot design using
  linear parameter-varying transformations,'' {\em J. Guidance, Control,
  Dynamics}, vol.~16, no.~2, pp.~256--261, 1993.

\bibitem{Tan97}
W.~Tan, {\em Applications of Linear Parameter-Varying Control Theory}.
\newblock PhD thesis, Dept. of Mechanical Engineering, University of California
  at Berkeley, 1997.

\bibitem{Efimov2016}
D.~Efimov and T.~Ra\"issi, ``Design of interval observers for uncertain
  dynamical systems,'' {\em Automation and Remote Control}, vol.~77, no.~2,
  pp.~191--225, 2016.

\bibitem{Raiessi2018}
T.~Ra\"issi and D.~Efimov, ``Some recent results on the design and
  implementation of interval observers for uncertain systems,'' {\em
  Automatisierungstechnik}, vol.~66, no.~3, pp.~213--224, 2018.

\bibitem{Chebotarev2015}
S.~Chebotarev, D.~Efimov, T.~Ra\"issi, and A.~Zolghadri, ``Interval observers
  for continuous-time {LPV} systems with $l_{1}/l_{2}$ performance,'' {\em
  Automatica}, vol.~58, no.~8, pp.~82--89, 2015.

\bibitem{Jaulin02}
L.~Jaulin, ``Nonlinear bounded-error state estimation of continuous time
  systems,'' {\em Automatica}, vol.~38, no.~2, pp.~1079--1082, 2002.

\bibitem{Kieffer_Walter04}
M.~Kieffer and E.~Walter, ``Guaranteed nonlinear state estimator for
  cooperative systems,'' {\em Numerical Algorithms}, vol.~37, pp.~187--198,
  2004.

\bibitem{Bernard_Gouze04}
B.~Olivier and J.~Gouz\'e, ``Closed loop observers bundle for uncertain
  biotechnological models,'' {\em Journal of Process Control}, vol.~14, no.~7,
  pp.~765--774, 2004.

\bibitem{Moisan_Bernard_Gouze09}
M.~Moisan, O.~Bernard, and J.~Gouz\'e, ``Near optimal interval observers bundle
  for uncertain bio-reactors,'' {\em Automatica}, vol.~45, no.~1, pp.~291--295,
  2009.

\bibitem{RVZ10}
T.~Ra\"issi, G.~Videau, and A.~Zolghadri, ``Interval observers design for
  consistency checks of nonlinear continuous-time systems,'' {\em Automatica},
  vol.~46, no.~3, pp.~518--527, 2010.

\bibitem{REZ11}
T.~Ra\"issi, D.~Efimov, and A.~Zolghadri, ``Interval state estimation for a
  class of nonlinear systems,'' {\em IEEE Trans. Automatic Control}, vol.~57,
  no.~1, pp.~260--265, 2012.

\bibitem{Efimov_a2012}
D.~Efimov, L.~Fridman, T.~Ra\"issi, A.~Zolghadri, and R.~Seydou, ``Interval
  estimation for {LPV} systems applying high order sliding mode techniques,''
  {\em Automatica}, vol.~48, pp.~2365--2371, 2012.

\bibitem{MazencBernard11}
F.~Mazenc and O.~Bernard, ``Interval observers for linear time-invariant
  systems with disturbances,'' {\em Automatica}, vol.~47, no.~1, pp.~140--147,
  2011.

\bibitem{Combastel2012}
C.~Combastel, ``Stable interval observers in {C} for linear systems with
  time-varying input bounds,'' {\em Automatic Control, IEEE Transactions on},
  vol.~PP, no.~99, pp.~1--6, 2013.

\bibitem{Efimov_a2013}
D.~Efimov, T.~Ra\"issi, S.~Chebotarev, and A.~Zolghadri, ``Interval state
  observer for nonlinear time varying systems,'' {\em Automatica}, vol.~49,
  no.~1, pp.~200--205, 2013.

\bibitem{Leurent2018}
E.~Leurent, Y.~Blanco, D.~Efimov, and O.-A. Maillard, ``Approximate robust
  control of uncertain dynamical systems,'' in {\em 32nd Conference on Neural
  Information Processing Systems (NeurIPS) MLITS Workshop}, (Montreal), 2018.

\bibitem{EFRZS12}
D.~Efimov, L.~Fridman, T.~Ra\"issi, A.~Zolghadri, and R.~Seydou, ``Interval
  estimation for {LPV} systems applying high order sliding mode techniques,''
  {\em Automatica}, vol.~48, pp.~2365--2371, 2012.

\bibitem{FarinaRinaldi2000}
L.~Farina and S.~Rinaldi, {\em Positive Linear Systems: Theory and
  Applications}.
\newblock New York: Wiley, 2000.

\bibitem{Smith95}
H.~Smith, {\em Monotone Dynamical Systems: An Introduction to the Theory of
  Competitive and Cooperative Systems}, vol.~41 of {\em Surveys and
  Monographs}.
\newblock Providence: AMS, 1995.

\bibitem{AitRami2008}
M.~Ait~Rami, C.~Cheng, and C.~de~Prada, ``Tight robust interval observers: an
  {LP} approach,'' in {\em Proc. of 47th IEEE Conference on Decision and
  Control}, (Cancun, Mexico), pp.~2967--2972, Dec. 9-11 2008.

\bibitem{Bolajraf2011}
M.~Bolajraf, M.~Ait~Rami, and U.~R. Helmke, ``Robust positive interval
  observers for uncertain positive systems,'' in {\em Proc. of the 18th IFAC
  World Congress}, pp.~14330--14334, 2011.

\bibitem{Efimov_tac2013}
D.~Efimov, T.~Ra\"issi, and A.~Zolghadri, ``Control of nonlinear and {LPV}
  systems: interval observer-based framework,'' {\em IEEE Trans. Automatic
  Control}, vol.~58, no.~3, pp.~773--782, 2013.

\bibitem{Efimov_CDC2019}
D.~Efimov and A.~Aleksandrov, ``Robust stability analysis and implementation of
  persidskii systems,'' in {\em Proc. IEEE Conference on Decision and Control
  (CDC)}, (Nice), 2019.

\bibitem{Khalil2002}
H.~K. Khalil, {\em Nonlinear Systems}.
\newblock Prentice Hall PTR, 3rd~ed., 2002.

\bibitem{highway-env}
E.~Leurent, ``An environment for autonomous driving decision-making.''
  \url{https://github.com/eleurent/highway-env}, 2018.

\bibitem{Polack2017}
P.~Polack, F.~Altch{\'{e}}, and B.~D'Andr{\'{e}}a-Novel, ``{The Kinematic
  Bicycle Model : a Consistent Model for Planning Feasible Trajectories for
  Autonomous Vehicles ?},'' {\em IEEEIntelligent Vehicles Symposium}, no.~Iv,
  pp.~6--8, 2017.

\bibitem{Treiber2000}
M.~Treiber, A.~Hennecke, and D.~Helbing, ``{Congested Traffic States in
  Empirical Observations and Microscopic Simulations},'' {\em Phys. Rev. E 62},
  2000.

\bibitem{Adrot2003}
O.~Adrot and J.-M. Flaus, ``{Trajectory computation of dynamic uncertain
  systems},'' {\em 42nd IEEE International Conference on Decision and Control
  (IEEE Cat. No.03CH37475)}, vol.~2, no.~December, pp.~1291--1296, 2003.

\end{thebibliography}

\end{document}